\documentclass[12pt]{article}

\usepackage{authblk}
\usepackage{amssymb}
\usepackage{amsmath}
\usepackage{mathrsfs}
\usepackage{amsfonts}
\usepackage{braket}
\usepackage{color}
\usepackage{graphicx}
\usepackage{tikz}
\usetikzlibrary{shapes,snakes}
\usepackage{pgfplots}
\pgfplotsset{compat=1.9} 
\newtheorem{thm}{Theorem}

\newtheorem{exmp}[thm]{Example}
\newtheorem{defn}[thm]{Definition}
\newcommand{\norm}[1]{\left\Vert#1\right\Vert}

\def\XXint#1#2#3{{\setbox0=\hbox{$#1{#2#3}{\int}$ }
\vcenter{\hbox{$#2#3$ }}\kern-.6125\wd0}}

\newcounter{lastnote}

\title{Covariant Quantum Fields via Lorentz Group Representation of Weyl Operators}
\author{Radhakrishnan Balu}
\affil{Army Research Laboratory Adelphi, MD, 21005-5069, USA \\
       radhakrishnan.balu.civ@mail.mil
      }  
\affil{Department of Mathematics \&
	Norbert Wiener Center for Harmonic Analysis and Applications, 
	University of Maryland,
	College Park, MD 20742.\\
	rbalu@umd.edu}          
\date{Received: date / Accepted: \today}
\begin{document}
\maketitle

\begin{abstract}
The building blocks of Hudson-Parthasarathy quantum stochastic calculus start with Weyl operators on a symmetric Fock space. To realize a relativistically covariant version of the calculus we construct representations of Poinca$\grave{r}$e group in terms of Weyl operators on suitably constructed, Bosonic or Fermionic based on the mass and spin of the fundamental particle, Fock spaces. We proceed by describing the orbits of homogeneous Lorentz group on $\mathscr{R}^4$ and build fiber bundle representations of Poinca$\grave{r}$e group induced from the stabilizer subgroups (little groups) and build the Boson Fock space of the Hilbert space formed from the sections of the bundle. Our Weyl operators are constructed on symmetric Fock space of this space and the corresponding annihilation, creation, and conservation operators are synthesized in the usual fashion in relativistic theories for space-like, time-like, and light-like fields. We achieve this by constructing transitive systems of imprimitivity (second-quantized SI), which are dynamical systems with trajectories dense in the configuration space, by induced representations. We provide the details of the field operators for the case of massive Bosons as the rest are similar in construction and indicate the ways to construct adapted processes paving way for building covariant quantum stochastic calculus.
\end{abstract}

\section{Introduction}
\label{intro}
Quantum stochastic calculus (QSC) \cite {KP1992} is a mathematically precise description of microscopic systems interacting with an environment that has been successfully applied to quantum optics based processes rigorously established by Accardi \cite {Accardi1990}, quantum filters framework of Belavkin \cite {Belavkin2005}, and quantum information processing framework of Mabuchi \cite {Tezak2012} to name a few. A relativistically covariant version of the program was developed by Applebaum \cite {Applebaum1995} for the Fermionic case and Frigerio for the Bosonic noise \cite {Frigerio1989}. Relativistic effects are an important consideration in Fermionic systems of topological quantum systems \cite {Witten2016}. Here, we present a systematic formulation of covariant QSC for all fundamental particles by extending the theory for relativistic quantum particles, that is based on systems of imprimitivity, to field theoretic context. Systems of imprimitivity is a covariant relationship between a pair of conjugate observables (can be thought of as {"}position{"} described by a projection measure and {"}momentum{"} operators by a unitary representation of the group) that can be used to describe a class of systems within unitary equivalence in a concise manner when there is a group representing the kinematics or dynamics of the system. Notions related relativistic quantum systems such as SI are based on orthodox quantum logic framework developed by Birkhoff-von Neumann and Mackey among others. Accordingly, our setup is separable Hilbert spaces that form standard quantum logics that are complemented orthomodular lattices \cite {Varadarajan1985}. The simplest example is to describe the Weyl commutation relation of a class of systems undergoing $Shr\ddot{o}dinger$ evolution. The SI in this case is defined with respect to the additive group of the real line to account for the time translation symmetry of the dynamics where the position operator is  a canonical observable. Hahn-Hellinger theorem \cite {KP1992} enables decomposittion of an arbitrary observable in terms canonical ones and so the SI approach is applicable to arbitrary quantum system and can be thought of as a quantization scheme. Some early examples for finite groups appeared in the works of von Neumann \cite{vN1949} that are ergodic systems. Wigner used these notions informally in his study on the theory of representations of the inhomogeneous Lorentz (Poincar$\grave{e}$) group \cite {WignerLz1939} that are transitive systems and Mackey \cite{Mackey1963} provided the systematic development of SI for systems with infinite degrees of freedom. 

Systems of imprimitivity are a very useful characterization of dynamical systems, when the configuration space of a quantum system is described by a group, from which infinitesimal forms in terms of differential equations ($Shr\ddot{o}dinger$, Heisenberg, and Dirac etc), and the canonical commutation relations can be derived. For example, using SI arguments it can be shown \cite {Wigner1949} that massless elementary particles with spin less than or equal to 1 can't have well defined position operators, that is photons are not localizable. The concept of localization, where the position operator is properly defined in a manifold, and covariance in relativistic sense of systems can be shown to be a consequence of systems of imprimitivity. Now, we have a procedure called Mackey{'}s machinery to synthesize SI, a class of unitarily equivalent systems, using representations induced by subgroups which in our primary example will be the little groups of inhomogeneous Lorentz group.  Our previous work in SI involved application to infinite unitary representations of the Poincar$\grave{e}$ group, that are induced by stabilizer subgroups or little groups in physics literature, in the context of quantum walks \cite {Rad2018} and \cite {Rad2019}. We follow the notions and notations from the work of Varadarajan \cite {Varadarajan1985} where the Dirac equation is derived in the most rigorous fashion as a consequence of SI. Varadarajan set up this program first by establishing standard quantum logics and geometry on orthomodular lattices and further developed the configuration space, of a Dirac particle, endowed with a topology, measure space, and a G-space of a locally compact second countable (lcsc) group G (for example the Poincar$\grave{e}$ group) providing a comprehensive view of the subject that exploits symmetry and is formulated in geometric terms. We extend this program by building the respective Fock spaces and then identifying Weyl operators indexed by elements of Poincar$\grave{e}$ group leading to covariant field operators (creation, annihilation, and preservation).

Let us first define the notion of SI and an important theorem by Mackey that characterizes such systems in terms of induced representations.

\begin {defn} \cite {Varadarajan1985} A G-space of a Borel group G is a Borel space X with a Borel automorphism $\forall{g\in{G}},t_g:x\rightarrow{g.x},x\in{X}$ such that
\begin {align}
&t_e \text{ is an identity} \\
&t_{g_1,g_2} =t_{g_1}t_{g_2}
\end {align}
The group G acts on X transitively if $\forall{x,y\in{X}},\exists{g}\in{G}\ni{x=g.y}.$
\end {defn}
\begin {defn} \cite {Varadarajan1985} A system of imprimitivity for a group G acting on a Hilbert space $\mathscr{H}$ is a pair (U, P) where $P: E\rightarrow{P_E}$ is a projection valued measure defined on the Borel space X with projections defined on the Hilbert space and U is a representation of G satisfying
\begin {equation}
U_gP_EU^{-1}_g = P_{g.E}
\end {equation}  
\end {defn} 
The next step is to consider semidirect product of groups, that naturally describe the dynamics system of relativistic quantum particles, and use the representation of the subgroup to induce a representation in such a way that it is an SI.

\begin {defn}
Let A and H be two groups and for each $h\in{H}$ let $t_{h}:a\rightarrow{h[a]}$ be an automorphism (defined below) of the group A. Further, we assume that $h\rightarrow{t_h}$ is a homomorphism of H into the group of automorphisms of A so that
\begin {align} \label{semidirectEq}
h[a] &= hah^{-1}, \forall{a\in{A}}. \\
h &= e_H, \text{  the identity element of H}. \\
t_{{h_1}{h_2}} &= t_{h_1}t_{h_2}.
\end {align}
Now, $G=H\rtimes{A}$ is a group with the multiplication rule of $(h_1,a_1)(h_2,a_2) = (h_1{h_2},a_1{t_{h_1}}[a_2])$. The identity element is $(e_H,e_A)$ and the inverse is given by $(h,a)^{-1} = (h^{-1},h^{-1}[a^{-1}]$. When H is the homogeneous Lorentz group and A is $R^4$ we get the Poincar$\grave{e}$ group via this construction.
\end {defn}

\begin {defn} Let group G be lcsc and X be a standard Borel G-space, and M be a standard Borel group. A map f is a (G, X, M)-cocycle relative to $\alpha$ which is an invariant measure class if (1) f is a Borel map of $G \times X$ into M. (2) f(e, x) = 1 for all $x \in X$. (3) $f(g_1 g_2, x) = f(g_1, g_2x) f(g_2, x)$ for almost all $(g_1, g_2, x) \in G \times G \times X$. If instead the map f satisfies conditions (2) and (3) for all x it is called a strict cocycle. 
\end {defn}
As an example let us construct a pair of observables obeying the SI relation.
\begin {exmp} \cite{Varadarajan1985} Let $\mathbb{G}$ be lcsc and X is the G-space which is Borel and endowed with a $\sigma-$finite measure $\alpha$. Then, we can define another $\sigma-$finite measure $\alpha^{g^{-1}}$ which is absolutely continuous with respect to $\alpha$ giving rise to a Radon-Nikodym derivative $r_g = \frac{d\alpha}{d\alpha^{g^{-1}}}$.  unitary operators that satisfy the SI relation as follows:
\begin {align*}
\mathscr{H} &= \mathscr{L}^2(X, k, \alpha), \text {  where, k is a separable Hilbert space  }. \\
U_g f(x) &= \{r_g (g^{-1} x)\}^{\frac{1}{2}} f(g^{-1}x),  f\in\mathscr{H} \text {  Momentum like operator}.\\
P_E F &= \chi_E f \text { Position like operator}. \\
\end {align*}
\end {exmp}

Next, let us recollect few results \cite{Varadarajan1985} on systems of imprimivity that will help us establish our main result in this work.
\begin {thm} (Theorem 5.19) There are quasi-invariant Borel measures on $X = G/G_0$, where G is a lcsc group and $G_0$ is a closed subgroup with the map $\beta:G \rightarrow X = g.x, x \in X$. If $\lambda$ is a finite  quasi-invariant measure on G, and $\tilde{\lambda}$ is the finite measure on X defined by the equation $\tilde{\lambda}(A) = \lambda(\beta^{-1}(A))$, then $\tilde{\lambda}$ is is quasi-invariant. Any two quasi-invariant $\sigma$-finite measures on X are mutually absolutely continuous. If $E \subset X$ is any Borel set, the quasi-invariant quasi-invariant $\sigma$-finite measures on X vanish for E if and only if $\beta^{-1}(E)$ is a null set of G. If $\alpha$ is any Borel measure on X, $\alpha$ is quasi-invariant (invariant) if and only if $\alpha^0$ is a quasi-invariant (lef-invariant) measure on G. 
\end {thm}

\begin {thm} (Theorem 6.19) Let G be a connected, simply connected complex semisimple Lie group, and $G_0$ a maximal compact subgroup of G. Let $\alpha$ be an invariant Borel measure on $X = G/G_0$. Let m be an irreducible representation of $G_0$ in a Hilbert space $\mathscr{K}$, and let $m{'}(g \rightarrow m{'}(g))$ be the unique complex-analytic homomorphism of G into the group of invertible linear endomorphisms of $\mathscr{K}$such that $m{'}(h) = m(h)$ for all $h \in G_0$. Then there exists exactly one map $x \rightarrow \langle .,. \rangle_x$ of X into the set of positive definite linear products on $\mathscr{K} \times\mathscr{K}$ such that $\langle .,. \rangle_{x_0} = \langle .,. \rangle$ and
$\langle u,v \rangle_x = \langle m{'}(g)u, m{'}(g)v \rangle_{g.x}$ for all $(g,x) \in G \times X$. If $\mathscr{H}$ is the Hilbert space of Borel maps f of X into $\mathscr{K}$ such that $\norm{f}^2 = \int_X \langle f(x), f(x) \rangle_x d\alpha(x) < \infty$, and we define P and U by 
\begin {align*}
U_g f(x) &= m{'}(g) f(g^{-1}x),  f\in\mathscr{H} \text {  Momentum operator}.\\
P_E F &= \chi_E f \text { Position  operator}. \\
\end {align*},
then (U, P) is a system of imprimitivity equivalent to the system induced by m.
\end {thm}

\section {Little groups (stabilizer subgroups)}
Some of the different systems of imprimitivity that live on the orbits of the stabilizer subgroups are described below. It is good to keep in mind the picture that SI is an irreducible unitary representation of Poincar$\grave{e}$ group $\mathscr{P}^+$ induced from the representation of a subgroup such as $SO_3$, which is a subgroup of homogeneous Lorentz, as $(U_m(g)\psi)(k) = e^{i\{k,g\}}\psi(R^{-1}_mk)$ where g belongs to the $\mathscr{R}^4$ portion of the Poincar$\grave{e}$ group, m is a member of the rotation group, $\{k,g\} = k_0 g_0 - k_1 g_1 - k_2 g_2 - k_3 g_3$, and the expression is in momentum space.

The stabilizer subgroups of the Poincar$\grave{e}$ group $\mathscr{P}^+$ can be described as follows: \cite {Kim1991}.
Time-like particle: There is a reference frame in which the 4-component momentum is proportional to (0,0,0,1) and the stabilizer subgroup is $SO_3$ that describes the spin. The walker, a massive particle, is rest in this frame. Let us denote the eigen vector of the Casimir operator $P_\mu P^\mu$ as  $\ket{0\lambda}$. Then, we can describe the invariant space under the Poincar$\grave{e}$ group $\mathscr{P}^+$ as $\Lambda_p\ket{0\lambda} = \ket{p\lambda}$ by applying the Lorentz boost. Any Lorentz operator operating on this space can be shown to be by a rotation (spin in our case).

Space-like particle: The Lorentz frame in which the walker is at rest has momentum proportional to (1,0,0,0) and the little group is
again $SO_3$ and this time the rotations will change the helicity. In this imaginary mass case the little group is rotations around the space axis and the analysis above carries through.

Light-like particle: There is no frame in which the relativistic quantum particle is at rest but the frame where the momentum is proportional to $(1,1,0,0)$ has the stabilizer subgroup with elements of the form $J_1$, that is a rotation around the first component of momentum and the boost $\Lambda_p$ in the spatial direction \cite {Kim1991}. These two operators commute and the induced representation can be constructed as above. 
\
\section {Fiber bundle representation of relativistic quantum particles}

The states of a freely evolving relativistic quantum particles are described by unitary irreducible representations of Poincar$\grave{e}$ group that has a geometric interpretation in terms of fiber bundles. 

The 3+1 spacetime discrete Lorentz group $\hat{O}(3,1)$-orbits  of the momentum space $\mathscr{R}^4$, where the systems of imprimitivity established will live, described by the symmetry $\hat{O}(3,1)\rtimes\mathbb{R}^4$. The orbits have an invariant measure $\alpha^+_m$ whose existence is guaranteed as the groups and the stabilizer groups concerned are unimodular and in fact it is the Lorentz invariant measure $\frac{dp}{p_0}$ for the case of forward mass hyperboloid.  The orbits are defined as:
\begin {align}
X^{+,2}_m &= \{p: p^2_0 - p^2_1 - p^2_2 - p^2_3 = m^2, p_0 > 0\}, \text{\color{blue} forward mass hyperboloid}.\\
X^{+,2}_m &= \{p: p^2_0 - p^2_1  - p^2_2 - p^2_3 = m^2, p_0 < 0\}, \text{\color{blue} backward mass hyperboloid}. \\
X_{00} &= \{0\}, \text{\color{blue} origin}.
\end {align}
Each of these orbits are invariant with respect to $\hat{O}(3,1)$ and let us consider the stabilizer subgroup of the first orbit at p=(0,0,0,1). Now, assuming that the spin of the particle is 1 (massive Boson) let us define the corresponding fiber bundles (vector) for the positive mass hyperboloid that corresponds to the positive-energy states by building the total space as a product of the orbits and $\mathbb{C}^4$.
\begin {align}
\hat{B}^{+,2}_m &= \{(p,v) \text{   }p\in{\hat{X}^+_m, }\text{   }v\in\mathscr{C}^2,\sum_k p_k\gamma_k v = mv\}. \\
\hat{\pi} &: (p,v) \rightarrow {p}. \text{  Projection from the total space }\hat{B}^{+,2}_m \text{ to the base }\hat{X}^{+,2}_m.
\end {align}
The states of the particles are defined on the Hilbert space $\hat{\mathscr{H}}^{+,2}_m$, square integrable functions on Borel sections of the bundle $\hat{B}^{+,2}_m$ with respect to the invariant measure $\alpha^+_m$, whose norm induced by the inner product is given below:
\begin {equation}
 \norm{\phi}^2 = \int_{X^+_m}p_0^{-1}\langle\phi{p},\phi{p}\rangle.{d\alpha}^+(p).
\end {equation}

To be more technical, we need Schwartz spaces, and their duals tempered distributions to guarantee Fourier transforms, to move with ease between configuration and momentum space descriptions. However, this is needed when setting up differential equations and so we will deal with them when we work with Fermionic examples in the future. 

Let us now state and establish the main result for the case of massive Boson with time-like momentum by constructing a strict cocycle from the representation of a subgroup following the prescription (lemma 5.24) in Varadarajan{'} text. The SI is a consequence of strict cocycle property and as he noted the construction is not canonical.

lemma 5.24 \cite{Varadarajan1985}: 
Let m be a Borel homomorphism of $G_0$ into M. Then there exists a Borel map b of G into M such that 
\begin {align} \label {eq: cocycle} 
B(e) &= 1. \\
b(gh) &= b(g)m(h), \forall (g,h) \in G \times G_0.  
\end {align}
Corresponding to any such map b, there is a unique stricy (G, X)-cocycle
f such that
\begin {equation} \label {eq: cocycle2}
f(g, g^1) = b(gg^1)b(g^1)^{-1}.
\end {equation}
$\forall (g,g^{-1}) \in G \times G$. f defines m at $x_0$. Conversely, f is a strict (G, X)-cocycle and b is a Borel map such that it satisfies \ref {eq: cocycle} equation pair, then the restriction of b to $G_0$ coincides with the homomorphism m defined of at $x_0$ and b satisfies \ref {eq: cocycle2}.

\begin {thm} Time-like Weyl representation of Poincar$\grave{e}$ group is a transitive system of imprimitivity. 
\end {thm}
\begin {proof}
Let $g:\mathscr{G} \rightarrow U_g(\hat{\mathscr{H}}^{+,2}_m)$ be a homomorphism from the rotation group $\mathscr{G} = SO_3$ to the unitary representation of the group in $\hat{\mathscr{H}}^{+,2}_m$. We note that it is a stabilizer subgroup which is also closed at the momentum (m,0,0,0) and so $H/\mathscr{G}$ is a transitive space.

Consider a map $v(g):\mathscr{G} \rightarrow \hat{\mathscr{H}}^{+,2}_m$ satisfying the first order cocycle relation $v(gh) = v(g) + U_g v(h), g,h \in \mathscr{G}$.
For an example of such a map consider the following: \cite {KP1992} 
\begin {align*}
\mathscr{H} &= \oplus_{j=0}^\infty \mathscr{H}_j. \\
H &= 1 \oplus \oplus_{j=1}^\infty H_j. \\
U_t &= e^{-itH}, t \in \mathscr{G}. \\
v(t) &= tu_0 \oplus \oplus_{j=1}^\infty (e^{-itH_j} u_j
- u_j).
\end {align*}
Now, we can define the Weyl operator $V_g = W_g (v(g), U_g)$ where $g \in \mathscr{G}$ in the Fock space $\Gamma_s(\hat{\mathscr{H}}^{+,2}_m)$. 

This is a projective unitary representation satisfying the commutator relation $V_g V_h = e^{iIm\langle v_g, U_g v_h \rangle} V_h V_g$ and let us denote the homomorphism from $\mathscr{G}$ to $V_g$ by m. This guarantees a map b that satisfies $b(gh) = b(g)m(h), g \in \mathscr{P}, h \in \mathscr{G}$ and such map can be constructed by considering the map
$c(x \rightarrow c(x)$ as Borel section of $\mathscr{P} / \mathscr{G}$ (not a canonical choice) with $c(x_0) = e$. The map
$\beta$ maps $g \in \mathscr{P} \rightarrow g \mathscr{G}$
\begin {align}
a(g) &= c(\beta (g))^{-1}. \\
b(g) &= m(a(g)). 
\end {align}
Then the strict cocycle $\phi$ satisfies $\phi(g_1, g_2) = b(g_1 g_2)b(g_2)^{-1}$.

We can now set up the SI relation as follows:
Let us denote the unique invariant measure class of X, noting that it is lcsc, by $\mathbb{G}$ and a member of this class as $\alpha$ and the Radon-Nikodym derivative $r_g = \frac{d\alpha}{d\alpha^{g^{-1}}}$.
\begin {align}
U_g f(x) &= \{r_g (g^{-1} x)\}^{\frac{1}{2}} \phi(g, g^{-1}x) f(g^{-1}x),  f\in\mathscr{H} \text {  Momentum operator}.\\
P_E F &= \chi_E f \text { Position  operator}.
\end {align}

$\blacksquare$
\end {proof}

We can construct the conjugate pair of field operators for the Fock space $\Gamma_s(\hat{\mathscr{H}}^{+,2}_m))$ as follows:

Let $p_g$ be the stone generator for the family of operators $P_{gt,p}, g \in \mathscr{P}, t \in \mathbb{R}$ and q(g) = p(ig) and we get the creation and annihilation operators as $a(g)^\dag = \frac {1}{2} ( q(g) - i p(g)$ and  $a(g) = \frac {1}{2} ( q(g) + i p(g)$.

\section {Summary and conclusions}
We derived the covariant field operators using induced representations of groups and  expressed them in terms of systems of imprimitivity. We established the results for the massive Boson case by inducing a representation of Poincar$\grave{e}$ group frmm the subgroup that is a stabilizer at the momentum (1,0,0,0). In the next step we will develop the Fermionic QSC and detail the basic adopted processes in this case. This opens up the ways for treating topological quantum materials with relativistic electrons that are of importance to a specific form of quantum computation in a rigorous fashion.


\end{document}